\documentclass[journal]{IEEEtran}

\usepackage{amsmath,amssymb}
\usepackage{cite}
\usepackage{graphicx}
\usepackage{booktabs}
\usepackage{algorithm}
\usepackage{algorithmic}
\usepackage{xcolor}
\usepackage{tikz}
\usepackage{subfig}
\usepackage{stfloats} 
\usetikzlibrary{arrows.meta,positioning,shapes.geometric}

\makeatletter
\def\@IEEENORMtitlevspace{1.2\baselineskip}
\def\@IEEEMINtitlevspace{0.9\baselineskip}
\def\section{\@startsection{section}{1}{\z@}{1.8ex plus 1ex minus 0.8ex}%
{0.6ex plus 1ex minus 0ex}{\normalfont\normalsize\centering\scshape}}%
\def\subsection{\@startsection{subsection}{2}{\z@}{1.6ex plus 1ex minus 0.8ex}%
{0.5ex plus .5ex minus 0ex}{\normalfont\normalsize\itshape}}%
\makeatother

\newtheorem{proposition}{Proposition}
\newtheorem{assumption}{Assumption}

\begin{document}

\title{Uncertainty-Disentangled Probabilistic Stability Analysis in Wind Power Integrated Weak Grids}

\author{Samson S.~Yu,~\IEEEmembership{Senior Member,~IEEE,}
	and Yinsong Chen~\IEEEmembership{}
	\thanks{S. Yu and Y. Chen are with the School of Engineering, Deakin University, Australia
		(e-mail: samson.yu@deakin.edu.au; yinsong.chen@deakin.edu.au). 
		}
}

\maketitle

\begin{abstract}
Conventional probabilistic small-signal stability analysis (PSSSA) propagates a single forecast distribution, conflating irreducible weather randomness (aleatoric) with reducible forecast-model uncertainty (epistemic). This letter propagates a second-order renewable forecast through the modal-stability map via an independent \emph{germ} variable, separating the two contributions exactly in closed form by a disentangled polynomial chaos expansion (d-PCE). The split underpins a forecast-aware $(\alpha,\beta,\gamma)$ stability certificate whose conservative branch converges to its irreducible aleatoric limit at $O(N^{-1/2})$---making a failed certificate diagnostic: epistemic-dominated risk recovers with better data; aleatoric-dominated risk needs improvements of the physical control system.
\end{abstract}
\vspace{-0.2cm}

\section{Introduction}
In converter-integrated power systems, modal damping depends on the renewable operating point, where a probabilistic forecast represents as a second-order object with two layers of opposite operational meaning \cite{rueda2009,hullermeier2021}: irreducible aleatoric weather randomness and reducible epistemic forecast-model error. Established PSSSA methods \cite{rueda2009,ni2017} collapse these into one distribution and report only total damping dispersion; the risk a better forecast could remove is hidden. Mixed aleatoric--epistemic uncertainty quantification (UQ) via \emph{germ} expansions exists elsewhere \cite{hullermeier2021,ni2017} but has not been applied to modal stability or coupled to an operational certificate. 

This letter contributes: (i)~an exact aleatoric/epistemic variance split of the modal stability margin via disentangled polynomial chaos (d-PCE, \emph{Proposition~\ref{prop:dpce}}); (ii)~a forecast-aware $(\alpha,\beta,\gamma)$ certificate with two-level finite-sample confidence bounds; and (iii)~$O(N^{-1/2})$ convergence of the conservative certificate to its aleatoric limit (\emph{Proposition~\ref{prop:converge}}), making a failed certificate diagnostic.

Operationally, the certificate provides a risk guarantee with confidence under forecast uncertainty and, by disentangling aleatoric and epistemic contributions, indicates whether instability can be mitigated by improved forecasting or requires physical system intervention.

\section{Forecast-Aware Modal-Stability Problem}

\begin{figure}[t]
\centering
\resizebox{\columnwidth}{!}{%
\begin{tikzpicture}[font=\footnotesize,>={Stealth[scale=0.9]}]
  \node[draw,rounded corners,minimum width=15mm,minimum height=7mm] (conv) {GFL conv.};
  \node[circle,fill,inner sep=1.1pt,right=10mm of conv] (pcc) {};
  \node[draw,minimum width=11mm,minimum height=6mm,right=8mm of pcc] (z) {$Z_{\mathrm{th}}$};
  \node[draw,circle,minimum size=6mm,right=17mm of z] (grid) {$\sim$};
  \draw (conv) -- node[above]{$P_w$} (pcc);
  \draw (pcc) -- (z);
  \draw (z) -- node[below]{SCR$=1.2$} (grid);
  \node[below=0.5mm of pcc] {PCC};
  \node[below=0.5mm of grid] {grid};
  \node[below=0.5mm of conv] {wind farm};
  \node[draw,dashed,rounded corners,below=9mm of conv,minimum width=15mm] (fc) {forecast $(u,\varphi)$};
  \node[right=6mm of fc] (xi) {$P_w{=}\xi{=}T(u,\varphi)$};
  \node[right=6mm of xi] (A) {$A(\xi)$};
  \node[right=5mm of A] (lam) {$\lambda(\xi)$};
  \node[right=5mm of lam] (g) {$g(\xi)$};
  \draw[->] (fc) -- (xi);
  \draw[->] (xi) -- (A);
  \draw[->] (A) -- (lam);
  \draw[->] (lam) -- (g);
  \draw[->,dashed] (xi) -- (conv);
\end{tikzpicture}}
\caption{Weak-grid test system (top) and the forecast-to-margin chain (bottom): the forecast \emph{germ} $u$ and model posterior $\varphi$ set the operating input $P_w=\xi=T(u,\varphi)$, which fixes the state matrix $A(\xi)$, its eigenvalues $\lambda(\xi)$, and stability margin $g(\xi)$. PCC, point of common coupling; SCR, short-circuit ratio.}
\label{fig:system}
\end{figure}
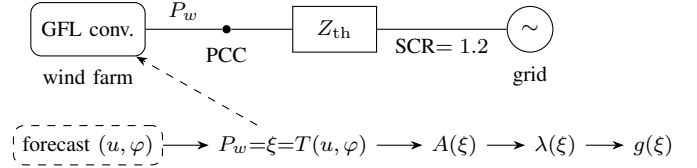

\subsection{Second-Order Forecast and Aleatoric Germ}
A probabilistic forecaster splits the wind operating driver $\xi$ ($P_w$) into two layers with opposite operational meaning \cite{chen2026}. The \emph{aleatoric} weather variance $\sigma_w^2$ is a stable climatological property, estimated once from an abundant historical record and treated as known and fixed (independent of the forecast training size). The \emph{epistemic} layer is the forecast model's uncertainty about the conditional \emph{mean} $\varphi=\mu(x_\star)$ at the forecast feature $x_\star$, given limited recent training data $D_N$, expressed as a posterior $\varphi\sim q(\varphi)=p(\varphi\mid D_N)$ that contracts as $N$ grows. The predictive law is $\xi\mid\varphi\sim\mathcal N(\varphi,\sigma_w^2)$, and conventional PSSSA is the degenerate case $q=\delta_{\varphi_0}$, a Dirac point mass at a single assumed mean $\varphi_0$.

To propagate this hierarchy we use an \emph{independent germ}, since the conditional law of $\xi$ depends on $\varphi$ and a direct tensor expansion in $(\xi,\varphi)$ would be invalid.
\begin{assumption}[Germ representation]\label{ass:germ}
There exists a standard-normal \emph{germ} $u\sim\mathcal N(0,1)$ with law $\rho(u)$ independent of $\varphi$, and a map $T$, such that $\xi=T(u,\varphi)$ has the predictive law. Here $T(u,\varphi)=\varphi+\sigma_w u$ with $\sigma_w$ the fixed climatological aleatoric scale; for a non-Gaussian predictive, $T(u,\varphi)=F_{\xi\mid\varphi}^{-1}(\Phi(u))$ is the conditional Rosenblatt map ($\Phi$ is the standard-normal cumulative distribution function, CDF).
\end{assumption}
Under \emph{Assumption~\ref{ass:germ}}, $u$ is perpendicular to $\varphi$, i.e., $u\perp\varphi$, ensuring the subsequent disentanglement exact. Measurement and process noise are carried by a second independent aleatoric \emph{germ} $u_\varepsilon\sim\mathcal N(0,1)$ with $\varepsilon=\sigma_\varepsilon u_\varepsilon$, where $\sigma_\varepsilon$ is the noise standard deviation.

\subsection{Modal Stability Margin}
We use a reduced-order grid-following converter model (Fig.~\ref{fig:system}) whose small-signal dynamics, linearised about the operating point fixed by $\xi$, are
\begin{equation}
	\Delta\dot{\mathbf x}=A(\xi)\,\Delta\mathbf x,\qquad \lambda_k(\xi)\in\mathrm{spec}\,A(\xi),
	\label{eq:state}
\end{equation}
where the state vector $\mathbf x$ collects the converter dynamic states (phase-locked-loop phase and frequency, current-control and output-filter states) and $A(\xi)$ is the state matrix at operating input $\xi$; the dominant weak-grid mode is the PLL-synchronisation mode, whose damping degrades as the short-circuit ratio falls~\cite{yang2021}. The propagation and certificate treat $A(\xi)$ as a black box, so any detailed electromagnetic-transient or library model can be substituted unchanged. For mode $k$, $\lambda_k(\xi)=\sigma_k+j\omega_k$ has damping ratio $\zeta_k=-\sigma_k/\sqrt{\sigma_k^2+\omega_k^2}$. Following standard practice, stability is an \emph{adequacy} criterion requiring all modes to satisfy minimum thresholds on damping ratio ($\zeta_{\min}$) and decay rate ($\eta_{\min}>0$), rather than enforcing the decay rate alone. To avoid dependence on an arbitrary scaling constant, the two criteria are expressed as \emph{dimensionless relative margins} and the binding margin is
\begin{equation}
	g(\xi)=\min_{k\in\mathcal K}\min\!\Big(\tfrac{\zeta_k-\zeta_{\min}}{\zeta_{\min}},\ \tfrac{-\sigma_k-\eta_{\min}}{\eta_{\min}}\Big),\quad \mathcal K=\{k:|\lambda_k|>0\},
	\label{eq:margin}
\end{equation}
where the inner minimum is over the two adequacy criteria (damping ratio and decay rate) and the outer over modes $k$, so the instability \emph{event} $\{g(\xi)+\varepsilon<0\}$ is invariant to the choice of normalising constants $\zeta_{\min},\eta_{\min}$. The realised margin is $\tilde m=g(\xi)+\varepsilon$; the map $g$ is deterministic, nonlinear and non-smooth (a minimum over modes).

\section{Disentangled Uncertainty Propagation}

\subsection{Exact Aleatoric/Epistemic Split via the Germ}
Write the structural margin as $G(u,\varphi)=g(T(u,\varphi))$ with $u\perp\varphi$. The law of total variance gives the exact identity
\begin{equation}
	\underbrace{\mathrm{Var}[G]}_{\mathrm{TU}}
	=\underbrace{\mathrm{Var}_{\varphi}\!\big(\mathbb{E}_u[G\mid\varphi]\big)}_{\mathrm{EU}}
	+\underbrace{\mathbb{E}_{\varphi}\!\big[\mathrm{Var}_u(G\mid\varphi)\big]}_{\mathrm{AU}}.
	\label{eq:ltv}
\end{equation}
Here $\mathbb E_u,\mathrm{Var}_u$ denote expectation and variance over the \emph{germ} $u$ at fixed $\varphi$, while $\mathbb E_\varphi,\mathrm{Var}_\varphi$ are taken over the posterior $q(\varphi)$ and $\mathrm{Var}[G]$ is the total variance over both. The margin's total uncertainty (TU) splits into the \emph{epistemic uncertainty} (EU)---the variance of the conditional mean across the posterior---and the \emph{aleatoric uncertainty} (AU)---the posterior-averaged conditional variance.

\subsection{Disentangled Polynomial Chaos Expansion (d-PCE)}
Let $\{\Psi_a(u)\}_{a\ge0}$ be orthonormal with respect to the fixed \emph{germ} law $\rho(u)$ ($\Psi_0\!=\!1$) and $\{\Phi_b(\varphi)\}_{b\ge0}$ orthonormal with respect to $q(\varphi)$ ($\Phi_0\!=\!1$). Because $u\perp\varphi$, the products $\Psi_a\Phi_b$ are orthonormal under the joint law, and
\begin{equation}
	G(u,\varphi)=\sum_{a,b}d_{ab}\,\Psi_a(u)\,\Phi_b(\varphi),
	\label{eq:dpce}
\end{equation}
with d-PCE coefficients $d_{ab}=\mathbb E[G\,\Psi_a(u)\Phi_b(\varphi)]$ and conditional moments $\mu(\varphi)\!=\!\mathbb E_u[G\mid\varphi]\!=\!\sum_b d_{0b}\Phi_b(\varphi)$ and $v(\varphi)\!=\!\mathrm{Var}_u(G\mid\varphi)\!=\!\sum_{a\ge1}(\sum_b d_{ab}\Phi_b(\varphi))^2$.

\begin{proposition}[Disentanglement]\label{prop:dpce}
Under \emph{Assumption~\ref{ass:germ}} and orthonormality, the converged expansion \eqref{eq:dpce} gives
\begin{equation}
	\mathrm{EU}=\!\sum_{b\ge1}\!d_{0b}^2,\quad
	\mathrm{AU}=\!\sum_{a\ge1}\sum_{b\ge0}\!d_{ab}^2,\quad
	\mathrm{TU}\!=\!\!\!\sum_{(a,b)\neq(0,0)}\!\!\!d_{ab}^2 .
	\label{eq:closedform}
\end{equation}
A third quantity, $\mathrm{Var}_\varphi(v(\varphi))$, is the epistemic uncertainty \emph{about the aleatoric spread}, distinct from both EU and AU and relevant to the certificate below.
\end{proposition}
\noindent\emph{Proof sketch.} Orthonormality and $u\perp\varphi$ give $\mathrm{EU}=\mathrm{Var}_\varphi(\sum_b d_{0b}\Phi_b)=\sum_{b\ge1}d_{0b}^2$ and $\mathrm{AU}=\mathbb E_\varphi[\sum_{a\ge1}(\sum_b d_{ab}\Phi_b)^2]=\sum_{a\ge1}\sum_b d_{ab}^2$. $\hfill\square$

\emph{Finite-ensemble realisation.} The spectral form \eqref{eq:closedform} requires fitting an orthonormal $\varphi$-basis. With a finite deep ensemble we instead use the equivalent \emph{empirical} estimator (\textbf{Algorithm~\ref{alg:cert}}): for each of the $M_\varphi$ posterior draws $\{\varphi_i\}$, we fit only the inner \emph{germ} expansion in $u$ and estimate the components by sample averaging over the ensemble,
\begin{equation}
	\widehat{\mathrm{AU}}=\tfrac1{M_\varphi}\!\sum_i v(\varphi_i),\quad
	\widehat{\mathrm{EU}}=\tfrac1{M_\varphi-1}\!\sum_i\big(\mu(\varphi_i)-\bar\mu\big)^2 ,
	\label{eq:empirical}
\end{equation}
with $\bar\mu$ the ensemble mean and $\widehat{\mathrm{Var}}_\varphi(v)=\tfrac1{M_\varphi-1}\sum_i(v(\varphi_i)-\bar v)^2$ ($\bar v$ the mean of $v(\varphi_i)$); the inner moments use $1$-$D$ \emph{germ} quadrature so $\widehat{\mathrm{EU}}$ has no nested-sampling bias.

\emph{Error control.} The expansion is truncated and fitted, so it carries truncation/regression error; on the non-smooth margin \eqref{eq:margin} a global expansion converges slowly, and we use a sparse least-angle-regression construction \cite{ni2017}. We report the leave-one-out error $e_{\mathrm{LOO}}$, set the surrogate certification margin $e_{\mathrm{sur}}$ to the maximum validation-grid residual, and cross-check $\widehat{\mathrm{AU}},\widehat{\mathrm{EU}}$ against an independent nested Monte~Carlo.

\subsection{Conditional Instability Probability}
The per-model (aleatoric) instability probability is the tail
\begin{equation}
	p_{\mathrm A}(\varphi)=\mathbb P_{u,u_\varepsilon}\big(G(u,\varphi)+\sigma_\varepsilon u_\varepsilon<0\big),
	\label{eq:pA}
\end{equation}
where $\mathbb P_{u,u_\varepsilon}$ is the probability over the germs $u$ and $u_\varepsilon$; this tail is \emph{not} determined by $\mu(\varphi),v(\varphi)$ alone and is evaluated by sampling the \emph{germs} through the cheap surrogate \eqref{eq:dpce}. A one-sided distribution-free upper bound (Cantelli) is available as a fast conservative screen,
\begin{equation}
	p_{\mathrm A}(\varphi)\le\frac{v(\varphi)+\sigma_\varepsilon^2}{v(\varphi)+\sigma_\varepsilon^2+\mu(\varphi)^2}\ \ \text{if }\mu(\varphi)>0,
	\label{eq:cantelli}
\end{equation}
the conservative certificate (Sec.~\ref{sec:certificate}) evaluates the tail with the sampled \eqref{eq:pA} for tightness, and can use the sampling-free screen \eqref{eq:cantelli} to certify easy models and preserve the $1-\gamma$ confidence.

\section{Stability Certificate and Forecasting Value}
\label{sec:certificate}

\subsection{Optimistic and Conservative Stability Certificates}
The predictive instability probability is $p_{\mathrm{pred}}=\mathbb E_\varphi[p_{\mathrm A}(\varphi)]$, which is estimated over the forecast posterior. 

We define two decision rules. 

\emph{Optimistic Certificate:} optimistically stable at risk $\alpha$ if $p_{\mathrm{pred}}\le\alpha$.

\emph{Conservative $(\alpha,\beta,\gamma)$ Certificate:} let the per-model safe event be $S(\varphi)=\{p_{\mathrm A}(\varphi)\le\alpha\}$, defined on the scale-invariant event \eqref{eq:pA}. From $M_\varphi$ posterior draws, $\hat p_{\mathrm{con}}=\tfrac1{M_\varphi}\sum_i\mathbf 1[\hat S(\varphi_i)]$, and we declare conservative stability when the one-sided Clopper--Pearson \emph{lower} bound on this fraction clears the target,
\begin{equation}
	\underline q\big(\hat p_{\mathrm{con}},M_\varphi;\,1-\gamma\big)\ \ge\ 1-\beta ,
	\label{eq:conservative}
\end{equation}
where $\underline q(\cdot)$ is the one-sided Clopper--Pearson lower confidence bound on the safe fraction. Two finite-sample caveats qualify the \emph{conservative certificate}: (i) \emph{Confidence:} exact evaluation of $\hat S$ (e.g.\ Cantelli screen \eqref{eq:cantelli}) gives confidence $1-\gamma$, while \emph{germ} sampling at $1-\beta_R$ per model yields overall confidence $\ge\!1-\gamma-M_\varphi\beta_R$ via union bound \cite{garatti2025}; and (ii) \emph{Surrogate margin:} certifying $\{\hat g+\varepsilon<e_{\mathrm{sur}}\}$ in \eqref{eq:pA} with the margin $e_{\mathrm{sur}}$ extends coverage toward the true $g$ ($e_{\mathrm{LOO}}$ is a fit diagnostic). The complete procedure is summarised in \textbf{Algorithm~\ref{alg:cert}}.

\begin{algorithm}[t]
	\small
	\caption{Disentangled Certification of Modal Stability}
	\label{alg:cert}
	\begin{algorithmic}[1]
		\vspace{-0.05cm}
		\STATE Fit probabilistic forecaster; obtain posterior draws $\{\varphi_i\}_{i=1}^{M_\varphi}$ and \emph{germ} map $T$.
		\STATE Per member, fit \emph{germ} d-PCE \eqref{eq:dpce}; record $e_{\mathrm{LOO}}$, set $e_{\mathrm{sur}}$.
		\STATE Read $\mu(\varphi_i),v(\varphi_i)$; form $\widehat{\mathrm{AU}},\widehat{\mathrm{EU}},\widehat{\mathrm{Var}}_\varphi(v)$ \eqref{eq:empirical}.
		\FOR{$i=1,\ldots,M_\varphi$}
		\STATE Sample \emph{germs}; estimate $p_{\mathrm A}(\varphi_i)$ \eqref{eq:pA} with margin $e_{\mathrm{sur}}$; set $\hat S(\varphi_i)$.
		\ENDFOR
		\STATE $p_{\mathrm{pred}}\!\gets\!\tfrac1{M_\varphi}\sum_i p_{\mathrm A}(\varphi_i)$;\ \ $\hat p_{\mathrm{con}}\!\gets\!\tfrac1{M_\varphi}\sum_i\mathbf 1[\hat S(\varphi_i)]$.
		\STATE Optimistically stable if $p_{\mathrm{pred}}\le\alpha$; conservatively stable if \eqref{eq:conservative} holds.
	\end{algorithmic}
		\vspace{-0.05cm}
\end{algorithm}
	\vspace{-0.1cm}

\subsection{Forecasting--Stability Coupling}
Let the certificate slack of model $\varphi$ be $r(\varphi)=\alpha-p_{\mathrm A}(\varphi)$, so $S(\varphi)=\{r(\varphi)\ge0\}$; the $\beta$-quantile $r^\beta=\inf\{t:\mathbb P_\varphi(r\le t)\ge\beta\}$ is what \eqref{eq:conservative} must clear. As $N\to\infty$ the posterior contracts, $q\to\delta_{\varphi^\star}$, and the conservative certificate converges to the aleatoric certificate $\{r(\varphi^\star)\ge0\}$, the irreducible floor set by the weather alone.

\begin{proposition}[Convergence and rate]\label{prop:converge}
Assume posterior contraction $\mathrm{Var}_q(\varphi)=O(1/N)$, with $O(\cdot)$ the asymptotic order as $N\to\infty$, and that $\varphi\mapsto p_{\mathrm A}(\varphi)$ is Lipschitz near $\varphi^\star$. Then $\mathrm{Var}_\varphi(r)=O(1/N)$ and the gap $\Delta(N)$ between the conservative threshold and its aleatoric limit closes as
\begin{equation}
	\Delta(N)=O\!\big(\mathrm{std}_\varphi(r)\big)=O\!\big(N^{-1/2}\big).
	\label{eq:rate}
\end{equation}
\end{proposition}
\noindent\emph{Proof sketch.} Lipschitz $p_{\mathrm A}$ gives $\mathrm{Var}_\varphi(r)=\mathrm{Var}_\varphi(p_{\mathrm A})=O(\mathrm{Var}_q(\varphi))=O(1/N)$; the robust quantile differs from the $\varphi^\star$ value by $O(\mathrm{std}_\varphi(r))$. $\hfill\square$

The gap is governed by $\mathrm{Var}_\varphi(r)$, which reflects variability in \emph{both} the conditional mean and margin ($\mathrm{Var}_\varphi(v)$ can be nonzero even when $\mathrm{EU}=0$). When AU is the dominant contributor, the risk is intrinsic and requires physical mitigation, such as damping support or curtailment.

\section{Case Studies}
\begin{figure*}[!b]
	\centering
	\subfloat[Decomposition at $N=50$]{\includegraphics[height=1.3in]{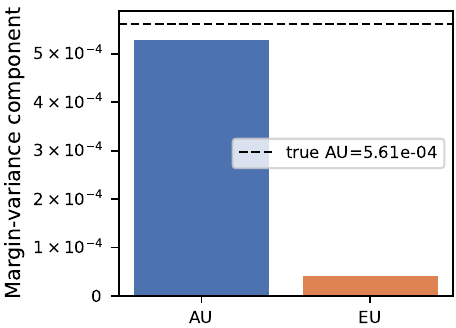}}
	\hfill
	\subfloat[EU shrinks, AU fixed]{\includegraphics[height=1.3in]{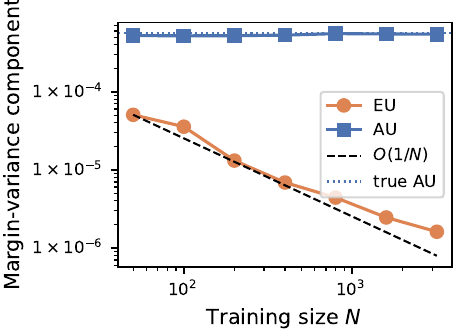}}
	\hfill
	\subfloat[Certificate slack contracts]{\includegraphics[height=1.3in]{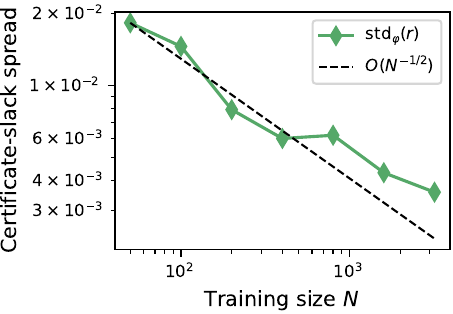}}
	\caption{Results, $8$-seed average. (a) At $N=50$, AU matches the ground-truth (dashed) within $5.7\%$; EU is the reducible $7.4\%$. (b) EU shrinks $\sim O(1/N)$ while AU stays at ground truth. (c) Slack spread $\mathrm{std}_\varphi(r)\sim O(N^{-1/2})$ (\emph{Proposition~\ref{prop:converge}}).}
	\label{fig:results}
	\vspace{-0.3cm}
\end{figure*}

\subsection{Test System and Case-Study Settings}
We use a reduced-order phase-locked-loop (PLL) grid-following converter (${\approx}28$~Hz) on a Th\'evenin grid. Wind is synthesised from a \emph{known} law (ground-truth $\sigma_w$ available); $\sigma_w$ is held fixed from climatological regression and the epistemic posterior is a $40$-member deep ensemble---sized for the outer Clopper--Pearson level, which needs $M_\varphi\gtrsim30$ to clear $1-\beta$ at confidence $1-\gamma$. We adopt a stringent damping-adequacy target $\zeta_{\min}=0.27$ and place the operating point just inside it, so wind variability yields a non-trivial instability probability. Table~\ref{tab:settings} lists the parameters.
\begin{table}[t]
	\centering
	\small
	\caption{Case study settings.}
	\label{tab:settings}
	\begin{tabular}{ll}
		\hline
		Grid model/SCR & Th\'evenin, SCR=$1.2$\\
		Damping-adequacy target & $\zeta_{\min}=0.27$ \\
		Decay-rate target & $\eta_{\min}=5.0$ \\
		Wind mean / $\sigma_w$ & $0.71$ / $0.071$ pu \\
		Process noise $\sigma_\varepsilon$ & $0.01$ \\
		Ensemble size $M_\varphi$ & $40$ \\
		$\alpha$ / $\beta$ / $1-\gamma$ & $0.10$ / $0.10$ / $0.95$ \\
		d-PCE order / $n_u$ & $8$--$10$ / $40000$ \\
		\hline
	\end{tabular}
\end{table}

\subsection{Study A: Disentanglement and Validation}
At $N=50$ ($8$-seed average), $\widehat{\mathrm{AU}}=5.3\times10^{-4}$ recovers the ground-truth $5.6\times10^{-4}$ within $5.7\%$ (Fig.~\ref{fig:results}(a)), with $e_{\mathrm{sur}}=4.9\times10^{-3}$; EU/TU$\,{=}\,7.4\%$ is the reducible share. The system is optimistically stable ($p_{\mathrm{pred}}=0.036\le\alpha$), but the Clopper--Pearson lower bound with $M_\varphi=40$ does not clear $0.90$ at $95\%$ confidence, motivating the data sweep in Study~B. 

\subsection{Study B: Forecast--Coupled Stability}
Fig.~\ref{fig:results}~(b) and~(c) sweep $N$ training data samples. EU falls from $5.1\times10^{-5}$ ($N=50$) to $1.6\times10^{-6}$ ($N=3200$) tracking $O(1/N)$, while AU holds at the ground truth. A run passes when the Clopper--Pearson bound \eqref{eq:conservative} clears $1-\beta$ at confidence $1-\gamma$. By this standard, the certificate flips fail$\to$pass as $N$ grows, with the across-seed pass rate rising $0.25\to0.75\to1.0$, and $\mathrm{std}_\varphi(r)$ contracts from $1.8\times10^{-2}$ to $3.6\times10^{-3}$, consistent with $O(N^{-1/2})$ (\emph{Proposition~\ref{prop:converge}}). This indicates that when EU dominates, better forecasting can restore certification; when AU dominates, improvements on physical systems are required, such as more effective control or load curtailment.

\section{Conclusion}
Using an independent \emph{germ}, the proposed d-PCE offers an exact aleatoric–epistemic split of the modal stability margin, enabling a conservative $(\alpha,\beta,\gamma)$ certificate with finite-sample confidence that converges to the irreducible aleatoric limit at $O(N^{-1/2})$.  The resulting certificate provides a stability assessment with confidence and indicates  whether instability is due to forecast uncertainty at the operating point or must be addressed through physical system measures.

\end{document}